\newif\ifacm
\acmfalse

\ifacm
	\documentclass[sigconf]{acmart}
\else
	\documentclass{article}
	\usepackage[vmargin=1in,hmargin={1in,1.5in},headsep=0.3in,headheight=81pt]{geometry}
	\usepackage[utf8]{inputenc}
	\usepackage{amsmath, amssymb, amsthm}
	\usepackage{url}
	\usepackage{hyperref}
	
\fi

\usepackage{kbordermatrix}
\renewcommand{\kbldelim}{(}
\renewcommand{\kbrdelim}{)}

\usepackage{multirow}
\usepackage{tikz}

\usetikzlibrary{
	babel,
	patterns,
	calc,
	plotmarks,
	shapes
}

\newcommand{\axeslinlog}[6]{
	
	\pgfmathsetmacro{\offs}{0.5};
	\pgfmathsetmacro{\scale}{0.7};
	\draw [black,->] (#1,#3) -- (#2+\offs,#3);
	\draw [black,->] (#1,#3) -- (#1,#4+\offs);
	
	\foreach \i in {#1,...,#2} {
		\draw (\i,#3-0.2) -- (\i,#3+0.2);
		\pgfmathsetmacro{\v}{int(\i * #5)};
		\node[black,scale=\scale] at (\i,#3-0.5) {$\v$};
	}
	
	\foreach \j in {#3,...,#4} {
		\draw (#1-0.2,\j) -- (#1+0.2,\j);
		\pgfmathsetmacro{\v}{int(\j + #6 - 1)};
		\node[black,scale=\scale] at (#1-0.7,\j) {$10^{\v}$};
	}
}

\definecolor{cornflowerblue}{rgb}{0.39, 0.58, 0.93}
\definecolor{brickred}{rgb}{0.8, 0.25, 0.33}
\definecolor{ao}{rgb}{0.0, 0.5, 0.0}
\definecolor{amber}{rgb}{1.0, 0.75, 0.0}

\AtBeginDocument{%
  }

\ifacm
	\setcopyright{acmcopyright}
	\copyrightyear{2022}
	\acmYear{2022}
	\acmDOI{XXXXXXX.XXXXXXX}
	
	\acmConference[CheckMATE'22]{CheckMATE 2022}{November 07--11,
	  2022}{Los Angeles, CA}
	\acmPrice{15.00}
	\acmISBN{978-1-4503-XXXX-X/18/06}
\fi

\begin{document}

\title{Efficient Deobfuscation of Linear Mixed Boolean-Arithmetic Expressions
	\ifacm\else
	\footnote{Accepted for presentation at the \href{https://www.cct.lsu.edu/\~checkmate/}{CheckMATE} workshop on research of offensive and defensive techniques in the context of Man-At-The-End (MATE) attacks, co-located with \href{https://www.sigsac.org/ccs/CCS2022/}{ACM Conference on Computer and Communications Security (CCS) 2022}, November 7 - 11, 2022.}
	\fi
}

\ifacm
\else
	\author{Benjamin Reichenwallner \& Peter Meerwald-Stadler \\ Denuvo GmbH \\ Salzburg, Austria}
	\date{}
\fi

\newcommand{\NN}{\mathbb{N}}
\newcommand{\ZZ}{\mathbb{Z}}
\newcommand\xor{\mathbin{\char`\^}}
\renewcommand\xor{\mathbin{^\wedge}}

\newcommand{\uand}{\mathbin{\&}}
\newcommand{\uor}{\mathbin{|}}
\newcommand\uxor{\xor}
\newcommand\unot{\mathord{\sim}}

\ifacm\else
	\makeatletter
	\def\thm@space@setup{%
		\thm@preskip=.7cm
		\thm@postskip=\thm@preskip 
	}
	\makeatother
	
	\newtheorem{definition}{Definition}
	\newtheorem{theorem}{Theorem}
	\newtheorem{corollary}{Corollary}
\fi

\newcommand{\abstractContent}{
	Mixed Boolean-Arithmetic (MBA) expressions are frequently used for obfuscation. As they combine arithmetic as well as Boolean operations, neither arithmetic laws nor transformation rules for logical formulas can be applied to suitably complex expressions, making MBAs hard to simplify and solve.
	
	In 2019, Liu et al. demystified linear MBAs, leveraging a transformation between the set $B=\{0,1\}$ of bit values and the set $B^n$ of words of length $n\in\mathbb{N}$ for linear MBAs, originally introduced by Zhou et al. in 2007. With their \textit{MBA-Blast} and \textit{MBA-Solver} algorithms, they outperform existing tools noticably in terms of performance as well as ability to simplify of such MBAs.
	
	We propose a surprisingly simple algorithm called \textit{SiMBA} that improves upon MBA-Blast and MBA-Solver in that it can deobfuscate all linear MBAs, does not miss particularly simple solutions and takes only a fraction of their runtime.
}

\maketitle

\ifacm\else
	\begin{abstract}
		\abstractContent
	\end{abstract}
\fi

\section{Introduction}

Mixed Boolean-arithmetic (MBA) transformation is a currently popular technique for code obfuscation introduced in the year 2006 by Zhou et al.~\cite{zhou1, zhou}. Simple expressions such as constants are replaced by semantically equivalent mixed Boolean-arithmetic expressions in order to make (commonly binary) code harder understandable and thereby help hide secret information such as data and algorithms --- e.g., used for watermarking or license checks --- via introduction of exaggerated complexity. The goal is to rewrite passages which are easily identifiable in program code as well as in binaries into expressions which are not easily referrable to the original ones.

MBA expressions contain logical as well as arithmetic operations. Due to a bad interaction between those, they cannot be resolved straightforward using any established SAT solvers or mathematical tools which either concentrate on logical or on arithmetic expressions. However, with the rise of methods to generate MBAs, a variety of tools for their deobfuscation, i.e., simplification, or verification are being developed. They use various techniques such as pattern matching (e.g., \mbox{SSPAM~\cite{sspam}}), neural networks (e.g., NeuReduce~\cite{neureduce}), bit-blasting (e.g., Arybo~\cite{arybo}), stochastic program synthesis (e.g., Stoke~\cite{stoke}, Syntia~\cite{syntia} and Xyntia~\cite{xyntia}) or synthesis-based expression simplification (e.g., QSynth~\cite{qsynth} and msynth~\cite{msynth}).

Sufficiently complex MBAs are commonly generated by a rewriting technique, iteratively replacing subexpressions by equivalent, more complex ones using a codebook of MBA identities~\cite{loki}. Obviously, such a codebook can also be used for an attempt to simplify complex MBAs. Since in general MBA expressions cannot be expected to appear in a codebook right away, an SMT solver can be used for equivalence checks against the listed simpler MBAs. MBAs can further be hardened against deobfuscation by applying additional encoding, invertible functions or point functions~\cite{zhou, loki}. This may introduce large constants which pose problems to many deobfuscation tools.

Recently, a first algorithm for fully algebraic attacks has been developed, incorporated in the highly related tools \textit{MBA-Blast}~\cite{mba-blast} and \textit{MBA-Solver}~\cite{mba-solver}. These make large progress to the deobfuscation of especially so-called \textit{linear MBAs} as well as of \textit{polynomial MBAs} with small restrictions. As written in their related papers and supported by our own experiments, these tools outperform other existing tools significantly for these classes of MBAs when it comes to simplification success and runtime. They are based on an alternative method for generating MBAs relying on a transformation of linear MBAs between the $n$-bit space for any $n\in \NN$ and the $1$-bit space (in other words, between bitwise and logical expressions) first described by Zhou et al.~\cite{zhou} and basically reverse the process for deobfuscation.

Motivated by their approach, we contribute a related, but different  algorithm which we call \textit{SiMBA} (for \textit{Simple MBA Simplifier}) and which is characterized by its simplicity, its competitive performance, its genericness derived through an implementation which is independent of the number of variables and the MBAs' complexity as well as the detail that it avoids a full transformation to the $1$-bit space, hence does not require a decomposition of an input MBA into terms consisting of bitwise expressions and constant factors. Therefore it allows more flexibility concerning the structure of its input. Its implementation as well as our self-generated datasets for its evaluation are publicly available on Github~\cite{github}. Moreover, we give a step-by-step deduction of the fundamental theorems the mentioned peer tools are based on.

In the next section we provide a definition of linear and polynomial MBAs. Additionally we describe how MBAs equivalent to a specific target function, that may itself be a linear MBA, can be generated. On the one hand this is relevant for the generation of MBAs for experiments and on the other hand the theorems presented there also pave the way for our deobfuscation technique that is also used in variations by peer tools. We revisit the groundbreaking theorem of Zhou et al.~\cite{zhou} and provide clear verifications for statements used for the generation of linear MBAs which are aimed at substitute the zero constant, a bitwise expression or even another linear MBA as well as for their deobfuscation.

The most comparable peer tools \textit{MBA-Blast} and \textit{MBA-Solver} which provide efficient deobfuscation algorithms for polynomial MBAs are described in the following section. This section additionally gives an overview of our SiMBA algorithm and its differences to MBA-Blast and MBA-Solver, and provides a proof of its formal foundation. Additionally, we provide results of experiments comparing SiMBA to these tools.

Finally, we point out our main contributions and give an outlook on the simplification of nonlinear MBAs which may use our approach too after identification of its linear parts.

\section{Preliminaries}

\subsection{Linear Mixed Boolean-Arithmetic Expressions}

Mixed Boolean-arithmetic expressions mix Boolean expressions with arithmetic ones. As there is obviously a strong connection between logical and bitwise operations, the concept of Boolean expressions is often intermixed with that of bitwise ones. While logical operators basically operate on $B=\{0,1\}$, i.e., single bits, bitwise operations are equivalently applied to all bits of $n$-bit words in $B^n$ for any fixed $n \in \NN$. This link is crucial for recent findings about MBAs, e.g., a generation algorithm proposed by Zhou et al.~\cite{zhou} and the highly related deobfuscation tools MBA-Blast~\cite{mba-blast} and MBA-Solver~\cite{mba-solver}.

We prefer the notion of bitwise expressions which fits better to our context. Hence we use the operators $\&$ (bitwise conjunction), $\xor$ (bitwise exclusive disjunction), $\vert$ (bitwise inclusive disjunction) and $\sim$ (bitwise negation) rather than $\land$ (logical conjunction), $\oplus$ (logical exclusive disjunction), $\lor$ (logical inclusive disjunction) and $\lnot$ (logical negation).

Although we are mainly interested in linear MBAs, we state the  --- more general --- definition of a polynomial MBA first.

\begin{definition}\label{def:poly}
	Let $B = \{0,1\}$ and $n,t \in \NN$. A \textit{polynomial mixed Boolean-arithmetic expression (MBA)} with values in $B^n$ and $t$ variables is a function $e: \left(B^n\right)^t \to B^n$ of the form $$e\left(x_1,\ldots,x_t\right) = \sum_{i\in I} a_i \prod_{j\in J_i} e_{ij}\left(x_1,\ldots,x_t\right),$$ where $I \subset \NN$ and $J_i\subset \NN$, for $i\in I$, are index sets, $a_i \in B^n$ are constants and $e_{ij}$ are bitwise expressions of $x_1,\ldots,x_t$ for $j\in I_j$ and $i \in I$.
\end{definition}

A linear MBA is a special kind of a polynomial one where each term consists of only one bitwise expression and an optional constant factor.

\begin{definition}\label{def:linear}
	Let $B = \{0,1\}$ and $n,t \in \NN$. A \textit{linear mixed Boolean-arithmetic expression (MBA)} with values in $B^n$ and $t$ variables is a function $e: \left(B^n\right)^t \to B^n$ of the form $$e\left(x_1,\ldots,x_t\right) = \sum_{i\in I} a_i e_i\left(x_1,\ldots,x_t\right),$$ where $I \subset \NN$ is an index set, $a_i \in B^n$ are constants and $e_i$ are bitwise expressions of $x_1,\ldots,x_t$ for $i \in I$.
\end{definition}

These definitions' origin lies in the groundbreaking paper of Zhou et al.~\cite{zhou} in which MBAs are proposed as a promising technique for obfuscation since logical and arithmetic operators do not interact nicely enough for MBAs to be easily resolvable and not much theory on MBAs existed at that time. In that paper, they also propose a much noticed method for generating linear MBAs. It relies on a fundamental relation between Boolean and bitwise expressions.

\begin{theorem}\label{thm:zhou}
	Let $n, s, t \in \NN$, $x_i$ variables over $B^n$ for $i=1,\ldots,t$ and $e_j: \left(B^n\right)^t \to B^n$ bitwise expressions on these variables for $j=1, \ldots, s$. Let $$e\left(x_1,\ldots,x_t\right) = \sum_{j=1}^s a_j e_j\left(x_1,\ldots,x_t\right)$$ be a linear combination of these bitwise expressions with coefficients $a_j \in B^n$ for $j=1,\ldots,s$ and hence a linear MBA. Furthermore let, again for $j=1,\ldots,s$, $\overline{e}_j: B^t \to B$ be the logical expression corresponding to $e_j$. Enumerate the possible combinations of zeros and ones for the variables by $B_t = \{b_1, \ldots, b_{2^t}\}$ arbitrarily, but fixed, and let $A = \left(v_{ij}\right) \in B^{2^t \times s}$ be the matrix of truth values of the $\overline{e}_j$'s with $v_{ij} = \overline{e}_j(b_i)$.
	
	Then $e \equiv 0$ if and only if the vector $Y=Y_a = (a_1,\ldots,a_s)^T$ is a solution of the linear equation system $AY = o$ over $B^n$, where $o = (0,\ldots,0)^T$ is the zero vector in $B^{2^t}$.
\end{theorem}

Note that the theorem's original statement only requires that the linear equation system is solvable, but this does not suffice in general. The concrete vector is required to be a solution. As we will see, this theorem is very crucial for the generation of linear MBAs as well as for their deobfuscation, but the latter purpose remained unnoticed for a long time, possibly due to its incorrect formulation.

Apart from that, Zhou et al.~\cite{zhou} also show that each bitwise expression $e$ of an arbitrary number $t\in \NN$ of variables has a nontrivial linear MBA representation $e = \sum_{j=1}^{2^t} a_j e_j$ with bitwise expressions $e_j\neq e$ and constants $a_j$ for $j=1,\ldots,2^t$. Moreover, there are infinitely many MBAs in total.

\subsection{Generation of Linear Mixed Boolean-Arithmetic Expressions}

Theorem~\ref{thm:zhou} implies a method for generating linear MBAs which is very contrary to the codebook approach proposed in other sources. While it suffers a bit from the very specific shape which generated expressions will have, its advantages are that it can choose from an infinite set of MBAs, and a linear MBA always exists for a given input.

For a fixed number $t$ of variables and another fixed number $s$ of bitwise expressions, any matrix $A \in B^{2^t \times s}$ containing only zeros and ones can be used for generating a linear MBA for $t$ variables over $B^n$ for some $n\in \NN$ if its columns are linearly dependent.

For given bitwise expressions $e_1,\ldots, e_s$ and its truth value matrix $A$ with linearly dependent columns, a solution of the linear equation system $AY=o$ gives the constants $a_j$ for a linear MBA $e = \sum_{j=1}^s a_j e_j$ that evaluates to zero for all possible inputs. As a consequence, any of those $s$ bitwise expressions, say, $e_j$ for some $j \in \{1,\ldots,s\}$, can be expressed as a linear combination of the other ones by multiplying $e$ by $a_j$'s multiplicative inverse in $B^n$. Note that if $A$'s columns were linearly independent, the only solution of the equation system would be the zero vector, which is of course not constructive.


Alternatively, one may generate all or a subset of $A$'s columns randomly and construct bitwise expressions fitting the corresponding truth values. This will add more variety, but if all bitwise expressions are generated randomly, it is in general not possible to find an MBA representing a bitwise expression (in place of zero) easily.


This is where the following extension of Theorem~\ref{thm:zhou} comes in handy.

\begin{corollary}\label{cor:ext}
	Let $n,s,t \in \NN$ and let the bitwise expressions $e_1,\ldots,$ $e_s$ and $e$ as well as the logical expressions $\overline{e}_1,\ldots,\overline{e}_s$ and the matrix $A$ be defined as in Theorem~\ref{thm:zhou}. Furthermore, let $f$ be another bitwise expression of $t$ variables and let $\overline{f}$ be its logical equivalent.
	
	Then $e \equiv f$ if and only if $AY_a=F$, where $Y_a = (a_1,\ldots, a_s)^T$ is the vector of $e$'s coefficients and $F=\left(\overline{f}(b_1),\ldots,\overline{f}(b_{2^t})\right)^T$ is the vector of $f$'s evaluations on all possible combinations of zeros and ones for the variables, restricted to one bit.
\end{corollary}

\begin{proof}
	If $e \equiv f$, then $e_1 :=e-f\equiv 0$. Since $e_1$'s truth table matrix is given by $A_1=(A, F)$, i.e., $e$'s truth table matrix with an appended column containing $f$'s truth values, and its coefficient vector is $Y_1 = (a_1,\ldots, a_s,-1)$, it follows from Theorem~\ref{thm:zhou} that $A_1Y_1 = o$ with $o= (0,\ldots,0)^T$. That is, $\sum_{j=1}^s a_j \overline e_j(b_i) - \overline f(b_i) = 0$ for all $i=1,\ldots, 2^t.$ The correctness of this implication follows.
	
	The other one, i.e., that $AY_a=F$ implies that $e \equiv f$, is proven by performing these steps in reversed order.
\end{proof}

\ifacm\else \vspace{.7cm} \fi

This theorem allows one to find a linear MBA representing an arbitrary bitwise expression $f$ via a method equivalent to that described above.

Finally, we generalize this even more for affine integer combinations of bitwise expressions.

\begin{corollary}\label{cor:ext2}
	Let $n,s,t \in \NN$ and let the bitwise expressions $e_1,\ldots,e_s$ and $e$ as well as the logical expressions $\overline{e}_1,\ldots,\overline{e}_s$ and the matrix $A$ be defined as in Theorem~\ref{thm:zhou}. Furthermore, let, for some $m\in\NN$, $f=\sum_{k=1}^m \alpha_k f_k + \beta$ be an affine combination of bitwise expressions of $t$ variables with $\alpha_1,\ldots,\alpha_m,\beta \in B^n$. Additionally, let $g=\sum_{k=1}^m \alpha_k \overline f_k - \beta$, where we write $-\beta$ for an integer that sums up to zero with $\beta$ in the modular field $B^n$ and $\overline f_k$ is the logical equivalent to $f_k$ for $k=1,\ldots,m$.
	
	Then $e \equiv f$ if and only if $AY_a=F$, where $Y_a = (a_1,\ldots, a_s)^T$ is the vector of $e$'s coefficients and $F=\left(g(b_1),\ldots,g(b_{2^t})\right)^T$ is the vector of $g$'s evaluations on all possible combinations of zeros and ones for the variables.
\end{corollary}

\begin{proof}
	The proof works similarly as that of Corollary~\ref{cor:ext}. In order to apply Theorem~\ref{thm:zhou}, we want to express the constant term as a multiple of a bitwise expression which evaluates to the $n$-bit number which has ones everywhere, i.e., $\sum_{i=1}^n 2^{i-1}\times 1 = 2^n-1 \equiv -1 \mod 2^n$. That is, the coefficient $\beta$'s sign is reversed in the application of the theorem. The rest follows.
\end{proof}

\ifacm\else \vspace{.7cm} \fi

Corollary~\ref{cor:ext2} allows us to find arbitrarily complex MBAs with a nearly arbitrary number of variables representing a large variety of affine combinations of bitwise expressions. We list some simple examples with two variables here:

\begin{align*}
	&x+y \to 2\,((x \uand y) \uor (\unot x \uand \unot y)) -2\, (\unot x \uand y) \\ &\hspace{1.1cm}+ 3\,((\unot x \uand y) \uor (x \uand \unot y)) -2\cdot \unot y \\
	&3\,735\,936\,685\cdot \unot x \to -3\,735\,936\,685 \, (x \uand \unot y) + 3\,735\,936\,685 \, (y \uor \unot x) \\ &\hspace{1.1cm}+ 3\,735\,936\,685 \, ((\unot x \uand y) | (x \uand \unot y)) - 3\,735\,936\,685 \, y \\
	&3\,735\,936\,685\, (x\uxor y) + 49\,374 \to 3\,735\,911\,998 \cdot \unot x -24\,687 \, (x \uor \unot y) \\ &\hspace{1.1cm}- 3\,735\,936\,685 \, (y \uor \unot x) + 3\,735\,911\,998 \, (y \uor x)
\end{align*}

What the corollary additionally suggests is that an affine combination of bitwise expressions is a linear MBA; that is, there is no need for a definition of, say, affine MBAs. Against this background, the corollary is related to Theorem~1 in~\cite{mba-solver} while differently and probably more constructively proven. It states that two linear MBAs are equivalent if and only if their \textit{signature vectors} coincide. We restate it for the sake of completeness:

\begin{corollary}\label{cor:sign}
	Let $e_1$ and $e_2$ be linear MBAs over words of the same length $n$ with truth table matrices $A_1$ and $A_2$, resp., and coefficient vectors $Y_1$ and $Y_2$, resp. Then $e_1 \equiv e_2$ if and only if $A_1 Y_1 = A_2 Y_2$ or, equivalently, their linear combinations of logical expressions corresponding to their bitwise expressions evaluate to the same values for all $b_i$, $i = 1,\ldots,2^t$ as defined previously.
\end{corollary}

\section{Deobfuscation of Linear Mixed Boolean-Arithmetic Expressions}

The mixed Boolean-arithmetic transform technique is commonly considered to be a powerful obfuscation method due to the incompatibility of its operands. However, Theorem~\ref{thm:zhou} and specifically Corollary~\ref{cor:sign} prepare an algebraic method which may circumvent this issue. This method has been first elaborated and implemented in the tools \textit{MBA-Blast} and \textit{MBA-Solver} which outperform existing tools significantly regarding success frequency and runtime at the deobfuscation of MBAs that have a specific shape.

\subsection{MBA-Blast and MBA-Solver}

MBA-Blast and MBA-Solver are two highly related tools for efficient deobfuscation of MBAs. Both are available in a prototype state on Github~\cite{gh-mba-blast, gh-mba-solver}. While the former can simplify polynomial MBAs, the latter applies a more efficient approach for those and attempts to resolve nonpolynomial MBAs as well. However, for that it needs a hint of a polynomial MBA included in a nonpolynomial MBA, whose replacement by a variable would cause it to be polynomial, as well as an additional input, making it irrelevant for practical applications at least in this prototype state.

The tools as available on Github strongly suffer from the requirement for prior knowledge about MBAs to resolve. E.g., MBAs have to meet exactly the canonical representation used in Definition~\ref{def:poly} and the user has to know whether they are linear, polynomial or nonpolynomial (with additional required knowledge in the latter case). A check for MBAs to be linear or polynomial, resp., is not performed.

The tools are implemented for simplifying MBAs working on two to four variables. First of all, each space of bitwise expressions using $t \in \NN$ variables has a basis of $2^t$ expressions (e.g., $16$ expressions for $t=4$ variables) each bitwise expression can be represented as a linear combination of. 
The MBA-Blast algorithm starts by replacing each bitwise expression occuring in an input MBA by such a linear combination. This is done via their evaluation for all possible combinations of truth values for the variables and hence leveraging Theorem~\ref{thm:zhou}. Although this crucial finding dates back to the year 2007~\cite{zhou}, it has seemingly never been recognized to be readily usable for deobfuscation before due to its incorrect statement --- see~\cite{eyrolles} or also the MBA-Blast paper~\cite{mba-blast} whose authors claim to be the first to prove it.

In contrast to that, MBA-Solver computes a so-called \textit{signature vector} for a whole linear MBA leveraging Corollary~\ref{cor:sign}, again via an evaluation of its bitwise expressions for all possible combinations of truth values, but here those are immediately multiplied by the bitwise expressions' coefficients and summed up while the linear combinations have to be summed up with MBA-Blast in order to derive a linear combination for a whole MBA. Hence, MBA-Solver saves this step, and consequently significant runtime.

More involved simplification steps, including simplification using Python's \texttt{SymPy} package, are necessary for nonlinear MBAs. To the best of our understanding, different bases are used for MBA-Blast and MBA-Solver. While the former uses, e.g., for two variables the vector $(x, y, x\uand y, -1)$ as a basis, the latter uses $(\unot(x\uor y), \unot (x\uor\unot y), x\uand\unot y, x\uand y)$. The reason for this discrepancy is not known to the authors, but we assume that this is just an exemplary base and a full implementation of this prototype would, as indicated in~\cite{mba-solver}, try various bases and choose the simplest result.

Alternatively, if the resulting signature vector, now also computed for a whole linear MBA with MBA-Blast, is nice enough and corresponds to one single bitwise expression, this one is the simplification result. In this case this expression is found in a lookup table listing standard representations of bitwise expressions for all possible combinations of the variables' truth values.

While the step of finding a linear combination of basis expressions can be implemented genericly to work for an arbitrary number of variables, usage of a lookup table for finding single bitwise expressions is only possible for a low number of variables. Note that a lookup table for $t=4$ variables has $2^{2^4}=65\,536$ entries, and one for $t=5$ variables would have $2^{2^5}=4\,294\,967\,296$ entries.

Due to this restriction and since it cannot be expected in general that an MBA can be resolved to one single bitwise expression, the tools might miss the nicest solutions in many cases. However, this may be acceptable since their main purpose is to make expressions easier for verification using SMT solvers~\cite{gh-mba-solver} and linear combinations might be nice enough.

For input expressions which satisfy the requirements on their structure, MBA-Blast and MBA-Solver indeed clearly outperform other existing tools regarding success rate and runtime.

\subsection{Our Algorithm}\label{subsec:our}

\subsubsection{Overview}

As MBA-Blast and MBA-Solver, SiMBA is written in \texttt{Python}. Although it was developed without knowledge of MBA-Solver's exact technique, it uses a rather similar technique also involving a computation of what one may call a \textit{signature vector} in the first step. However, while MBA-Blast and MBA-Solver use Corollaries~\ref{cor:ext} and~\ref{cor:sign}, resp., and hence operate in the $1$-bit space $B=\{0,1\}$, we do not fully perform this transformation from $B^n$ (for $n\in\NN$) to $B$. Rather we immediately evaluate the whole linear MBA for the combinations of truth values and hence get vectors which live in $B^n$. The following theorem provides the formal basis for that.

\begin{theorem}\label{thm:main}
	Let $e$ and $f$ be linear MBAs over words of the same length $n$ and let $t\in\NN$ be their (maximum) number of variables. Then $e \equiv f$ if and only if $e(b_i) = f(b_i)$ for all $b_i \in B_t$, $i = 1,\ldots,2^t$, as defined previously. 
\end{theorem}

Note that the difference to Corollary~\ref{cor:sign} is that we evaluate the bitwise expressions rather than their logical ($1$-bit) equivalents for the $b_i$'s and hence get $n$-bit results. That implies that we do not have to decompose a linear MBA into its terms and factors, but we can evaluate it as a whole.

\ifacm\else \vspace{.7cm} \fi

\begin{proof}
	We only have to prove one direction since the other one is trivial. We assume that $e(b_i) = f(b_i)$ for all $b_i$, $i = 1,\ldots,2^t$ and write $e = \sum_{j=1}^s a_j e_j$ and $f = \sum_{j=1}^m \alpha_j f_j$ for $s,m \in \NN$, $a_1,\ldots, a_s,$ $\alpha_1,\ldots, \alpha_m \in B^n$ and bitwise expressions $e_1,\ldots, e_s,$ $f_1,\ldots, f_m$.
	
	We have for all $i \in \{1,\ldots,2^t\}$ that $\sum_{j=1}^s a_j e_j(b_i) = \sum_{j=1}^m \alpha_j f_j(b_i)$. Now let $\overline e_j$ be the logical equivalent of $e_j$ for $j=1,\ldots,s$ and $\overline f_j$ be the logical equivalent of $f_j$ for $j=1,\ldots,m$. Furthermore, for $i=1,\ldots,2^t$, we use the same notation for $b_i$ with one bit as for $b_i$ with $n$ bits, but we note that the $n$-bit version has zeros in the additional $n-1$ bits of its entries. 
	
	We want to show that $\sum_{j=1}^s a_j \overline e_j(b_i) = \sum_{j=1}^m \alpha_j \overline f_j(b_i)$ for all $i \in \{1,\ldots,2^t\}$. Then we may use Corollary~\ref{cor:sign}. We have the following for any $i\in \{1,\ldots,2^t\}$ after already interchanging the sums:
	
	\begin{align*}
		&\sum_{j=1}^s a_j \overline e_j(b_i) + \sum_{k=2}^n 2^{k-1} \sum_{j=1}^s a_j \overline e_j(o) \\ &\hspace{2cm} = \sum_{j=1}^m \alpha_j \overline f_j(b_i) + \sum_{k=2}^n 2^{k-1} \sum_{j=1}^m \alpha_j \overline f_j(o),
	\end{align*}
	where we write $o$ for the zero vector in $B^t$ and the sum $\sum_{k=2}^n 2^{k-1}$ is obviously constant ($2^n-2$, in more detail).

	If it holds that $\sum_{j=1}^s a_j \overline e_j(o) = \sum_{j=1}^m \alpha_j \overline f_j(o),$ the proof is finished. We consider the equation for the specific $b_{i^\prime}\in B_t$ which assigns zero to every variable:
	
	\begin{align*}
		&(2^n-1)\, \sum_{j=1}^s a_j \overline e_j(o) =  (2^n-1)\, \sum_{j=1}^m \alpha_j \overline f_j(o).
	\end{align*}
	This shows that $\sum_{j=1}^s a_j \overline e_j(o) = \sum_{j=1}^m \alpha_j \overline f_j(o)$.
\end{proof}

\ifacm\else \vspace{.7cm} \fi

Depending on the number of variables, the SiMBA algorithm is comprised of two or three steps described subsequently in more detail:
\begin{enumerate}
	\item For an input MBA $e$, a signature vector $F$ is derived by evaluating $e$ for all possible combinations of zeros and ones for the variables.
	\item Based on $F$, the input MBA is transformed into a linear combination of base bitwise expressions, using Theorem~\ref{thm:main}. This is in any case a valid, but not necessarily optimal solution.
	\item If the resulting linear combination uses fewer than three variables, various attempts to find a simpler expression via lookup tables are performed.
\end{enumerate}

The main advantage of our approach is that due to the immediate evaluation, one does not have to decompose an MBA into terms of bitwise expressions and their coefficients and hence does not require a concrete structure of input expressions. Optionally a check for the MBA's linearity can be performed. Moreover we expect a better performance especially for complex MBAs with a higher number of variables.

\subsubsection{Computing the Signature Vector}

In order to be able to use Theorem~\ref{thm:main}, we evaluate an input expression $e$ on all possible combinations of zeros and ones for the variables and obtain a vector $F=e(b_1,\ldots,b_{2^t}).$ In order to be most generic, for $k \in \{1,\ldots,2^t\}$ and $i \in \{1,\ldots,t\}$, $b_k$ assigns the value $1$ to the variable $x_i$ if $k$'s remainder after a division by $2^i$ is larger than $2^{i-1}$:
\begin{align*}
	b_1 &= (0,0,0,\ldots,0),\\
	b_2 &= (1,0,0,\ldots,0),\\
	b_3 &= (0,1,0,\ldots,0),\\
	b_4 &= (1,1,0,\ldots,0),\\
	&\hspace{5mm}\vdots\\
	b_{2^t} &= (1,1,1,\ldots,1).
\end{align*}

As an example, we consider the simple linear MBA $e(x,y) = 3\,735\,936\,685\,(x\uxor y) + 49\,374$. In $B$, we would get a vector $$(e^\prime(0,0), e^\prime(1,0), e^\prime(0,1), e^\prime(1,1)) = (0, 1, 1, 0)$$ for the bitwise part $e^\prime(x,y) = x\uxor y$. This vector would have to be multiplied by the coefficient and the constant would have to be subtracted (see Corollary~\ref{cor:ext2}) in order to obtain a vector $F_1^\prime$. In our setting, we immediately compute 
\begin{align*}F_1&=(e(0,0), e(1,0), e(0,1), e(1,1)) \\&= (49\,374, 3\,735\,986\,059, 3\,735\,986\,059, 49\,374) \\
	&= F_1^\prime + 2\times 49\,374.
\end{align*}

We see another difference when we consider the negation $f(x)=\unot x$. Rather than using the signature vector $$F_2^\prime=(f(0),f(1)) = (1, 0),$$ we use $$F_2 = (f(0),f(1)) = (-1,-2) = (-0-1, -1-1),$$ where we evaluate $f$ on $B$ in the former evaluation and on $B^n$ in the latter one.

\subsubsection{Finding a Linear Combination}

The next step is the derivation of a linear combination of basis elements for bitwise expressions with $t\in\NN$ variables. While in theory any basis can be used, we want to describe a very generic approach. If we enumerate the variables as $x_1,\ldots,x_t$, the following is a basis for $t$ variables: 
\begin{align*}
	(&1,  \\
	&x_1, \ldots, x_t, \\
	&x_1 \uand x_2, \ldots, x_1 \uand x_t, x_2 \uand x_3, \ldots, x_{2} \uand x_t, \ldots, x_{t-1} \uand x_t, \\
	&x_1 \uand x_2\uand x_3, x_1 \uand x_2\uand x_4, \ldots, x_{t-2}\uand x_{t-1}\uand x_t, \\
	&\vdots\\
	&x_{1} \uand \cdots \uand x_t). \\
\end{align*}

It has the following truth table $A$: \footnotesize
\begin{align*}
	\kbordermatrix{
		& 1 & x_1 & x_2 & \cdots & x_1 \uand x_2 & x_{1} \uand x_3 & \cdots &x_{1} \uand \cdots \uand x_t\\
		(0,0,0,\ldots,0) & 1 & 0 & 0 & \cdots & 0 & 0 & \cdots & 0 \\
		(1,0,0,\ldots,0) & 1 & 1 & 0 & \cdots & 0 & 0 & \cdots & 0 \\
		(0,1,0,\ldots,0) & 1 & 0 & 1 & \cdots & 0 & 0 & \cdots & 0 \\
		(1,1,0,\ldots,0) & 1 & 1 & 1 & \cdots & 1 & 0 & \cdots & 0 \\
		(0,0,1,\ldots,0) & 1 & 0 & 0 & \cdots & 0 & 0 & \cdots & 0 \\
		(1,0,1,\ldots,0) & 1 & 1 & 0 & \cdots & 0 & 1 & \cdots & 0 \\
		& \vdots & \vdots & \vdots & \ddots & \vdots & \vdots & \ddots & \vdots \\
		(1,1,1,\ldots,1) & 1 & 1 & 1 & \cdots & 1 & 1 & \cdots & 1
	}
\end{align*}
\normalsize

If we now want to solve a linear equation system $AY = F$, where $F$ is a vector derived in the first step, we do not need much linear algebra and no utilities. If we eliminate $A$'s columns and rows wisely, we permanently have a row which only has one $1$ in it until we have computed all coefficients in $Y$.

In more detail, the first row in which we have a $1$ for a variable $x_i$ for $i \in \{1,\ldots,t\}$ is the $(2^{i-1}+1)$-th one and the first row in which we have a $1$ for a conjunction $x_{i_1} \uand \cdots \uand x_{i_m}$ with $i_1 < \cdots < i_m$ for $i_j \in \{1,\ldots,t\}$, $j\in\{1,\ldots,m\}$, $m \in \NN$ is the row number $\sum_{j=1}^m 2^{i_j-1} + 1$. These numbers indicate $F$'s entries which we find conjunctions' coefficients in. Similarly, we constantly have to apply the corresponding manipulations for all of $F$'s entries with indices of rows in which $A$'s column has a $1$ too.

This approach cannot fail if we handle the conjunctions in an order depending on their numbers of variables, starting with the constant in row one. As soon as this is finished, we have a first viable solution in terms of a linear combination of base expressions.

Returning to the first example, $F_1$'s first entry $49\,374$ is the constant term which has to be subtracted from $F_1$'s other entries to get $(0, 3\,735\,936\,685, 3\,735\,936\,685, 0)$. Now, the coefficient for both variables --- we call them $x$ and $y$ again --- is $3\,735\,936\,685$, and the vector is modified to $(0, 0, 0, -2\times 3\,735\,936\,685)$. This implies the following linear combination:
\begin{align*}
	e(x,y) &\equiv 49\,374 + 3\,735\,936\,685\, x + 3\,735\,936\,685\, y \\ &- 7\,471\,873\,370\, (x\uand y).
\end{align*}

\subsubsection{Finding a Simpler Solution}

In order not to miss very simple results, we perform an alternative refinement attempt similar to that of MBA-Blast and MBA-Solver, but more involved. For that we use truth tables for $1 \leq t \leq 3$. If an input expression has more than three variables, but not all of them appear in the linear combination, we can reduce the number of variables and potentially perform this refinement step too.

As indicated before, we cannot identify a negation just by flipped bits, but we can use the identity $\unot x = -x-1.$ We perform the following attempts to find a result which has fewer terms than the linear combination in the following order using a lookup table and finish as soon as an attempt is successful:
\begin{enumerate}
	\item If all entries of $F$ coincide, we have a constant expression.
	\vspace{2mm}

	\textit{Example for $t=2$:} The vector $$(49\,374,49\,374,49\,374,49\,374)$$ obviously corresponds to $$e(x,y) = 49\,374.$$
	
	\item If $F$ has two unique entries and its first entry is zero, we replace the nonzero element $a$ by $1$, find the lookup table's entry for the corresponding truth vector and multiply the found expression by $a$.
	\vspace{2mm}
	
	\textit{Example for $t=2$:} The vector $$(0,49\,374,49\,374,0)$$ corresponds to $$e(x,y) = 49\,374\, (x\uxor y).$$
	
	\item If $F$ has two unique entries $a$ and $b$, both of them are nonzero, w.l.o.g., $b \equiv 2a \mod 2^n$, and $F$'s first entry is $a$, we can express the result in terms of a negated single expression. We replace all occurences of $a$ by zeros and that of $b$ by ones, find the corresponding expression in the lookup table, negate it, and multiply it by $-a$.
	\vspace{2mm}
	
	\textit{Example for $t=2$:} For the vector $$(49\,374, 98\,748, 98\,748, 49\,374),$$ we have $a=49\,374$ and $b=98\,748$. The vector $(0,1,1,0)$ corresponds to $x\uxor y$. Hence we have a result $$e(x,y) = -49\,374\cdot \unot(x\uxor y).$$
	
	\item If $F$ has two unique entries $a$ and $b$, but the previous cases do not apply, and $F$'s very first entry is $a$, we first identify $a$ as the constant term. Then we find an expression with ones exactly where $F$ has the entry $b$ in the lookup table, multiply it by $b-a$ and add the term to the constant.
	\vspace{2mm}
	
	\textit{Example for $t=2$:} For the vector $$(49\,374, 3\,735\,936\,685, 3\,735\,936\,685, 49\,374),$$ we have a constant $a=49\,374$ and, with $b=3\,735\,936\,685$, relate the vector $(0,b,b,0)$ to $x\uxor y$ and arrive at the expression $$e(x,y) =  49\,374 + 3\,735\,887\,311\,(x \uxor y).$$
	
	\item If $F$ has two unique nonzero entries $a$ and $b$ and its first one is zero, we split it into two vectors with ones where $F$ has entries $a$ or $b$, resp., find the corresponding expressions in the lookup table, multiply them by $a$ and $b$, resp., and add the terms together.
	\vspace{2mm}
	
	\textit{Example for $t=2$:} We split the vector $$(0, 3\,735\,936\,685, 3\,735\,936\,685, 49\,374)$$ into the vectors $(0,0,0,a)$ and $(0,b,b,0)$ with $a=49\,374$ and $b=3\,735\,936\,685$, relate the former to $x\uand y$ and the latter to $x\uxor y$ and arrive at the expression $$e(x,y)=49\,374\, (x\uand y) + 3\,735\,936\,685\, (x\uxor y).$$
	
	\item If $F$ has three unique nonzero entries $a$, $b$ and $c$ and its first one is $0$, we try to express one of them as a sum of the others modulo $2^n$, e.g., $a \equiv b + c.$ In that case we split $F$ into two vectors with ones where $F$ has entries $b$ or $c$, resp., or $a$, find the corresponding expressions in the lookup table, multiply them by $b$ and $c$, resp., and add the terms together.
	\vspace{2mm}
	
	\textit{Example for $t=2$:} We split the vector $$(0, 3\,735\,936\,685, 3\,735\,887\,311, 49\,374)$$ into the vectors $(0,a,a,0)$ and $(0,b,0,b)$ with  $a=3\,735\,887\,311$ and $b=49\,374$, relate the former to $x\uxor y$ and the latter to $x$ and arrive at the expression $$e(x,y)=3\,735\,887\,311\, (x\uxor y) + 49\,374\, x.$$
	
	\item If $F$ has three unique nonzero entries $a$, $b$ and $c$, its first one is $0$ and the previous case does not apply, we split it into three vectors with ones where $F$ has entries $a$, $b$ or $c$, resp., find the corresponding expressions in the lookup table, multiply them by $a$, $b$ and $c$, resp., and add the terms together.
	\vspace{2mm}
	
	\textit{Example for $t=2$:} We split the vector $$(0, 49\,374, 3\,735\,936\,685, 201)$$ into the vectors $(0,a,0,0)$, $(0,0,b,0)$ and $(0,0,0,c)$ with $a=49\,374$, $b=3\,735\,936\,685$ and $c=201$, relate them to $x\uand \unot y$, $\unot(x\uor\unot y)$ and $x\uand y$ and arrive at the expression 
	\begin{align*}
		e(x,y) &=49\,374\, (x\uand \unot y) + 3\,735\,936\,685\cdot \unot(x\uor\unot y) \\ &+ 201\, (x\uand y),
	\end{align*}
but we neglect it since it is not simpler than the linear combination $$49\,374\, x + 3\,735\,936\,685\, y \\ + 201\, (x\uand y).$$

	\item If $F$ has four unique nonzero values and its first one is nonzero, proceed as above after subtracting the first entry, i.e., the constant term, and add that to the final result.
	\vspace{2mm}
	
	\textit{Note for $t=2$:} That is not relevant for $t=2$ since the resulting term would never be simpler than the linear combination of base expressions, but it may be for $t=3$.
\end{enumerate}

This list of attempts can in theory be extended, but a linear combination of a higher number of bitwise terms from the lookup table might soon appear to be more complex than a linear combination of an even higher number of very simple base expressions. For $t=2$, the algorithm will definitely find a simplest result.

\subsection{Verification and Comparison}

In order to classify our algorithm \textit{SiMBA}, we perform four experiments. In the first, we use a dataset of $10\,000$ linear MBAs with a varying number of variables provided by the Github repository of \textit{NeuReduce}~\cite{gh-neureduce} to verify correctness and universality.

Thereafter, we compare the algorithm with the most relevant peer tools. We can confirm the results of~\cite{mba-blast} and~\cite{mba-solver} that Arybo, SSPAM and Syntia do not reliably simplify MBAs in general and, if they do, need significantly more time. Hence, we are content with a comparison to MBA-Blast and MBA-Solver only. Corresponding to those, we expect that they do similar things, but MBA-Solver is more efficient since it saves a simplification step. Especially for expressions which can be simplified to one single bitwise expression, the first step is omitted completely.

In the second experiment, we use MBA-Blast (partially), MBA-Solver and SiMBA to simplify all linear MBAs provided by the former's Github repository~\cite{gh-mba-solver}. In the third one, we run experiments on various MBAs that we generated for simple affine combinations of bitwise expressions using Corollary~\ref{cor:ext2}. The final experiment indicates that an affine output encoding does not effect SiMBA's success rate and runtime.

All experiments are performed on a Debian Bullseye virtual machine run on an Intel Core i7-12700K CPU and $3.6$ GHz. The runtime was measured using \texttt{Python 3.9} with the \texttt{time} package. Furthermore, we use $n=64$ bits in all experiments.

\subsubsection{Verification on NeuReduce's Dataset}

In a first experiment, we run SiMBA on NeuReduce's test dataset containing $10\,000$ expressions with two, three, four or five variables. 

\begin{table}[h]
	\ifacm\else
		\centering
	\fi
	\begin{tabular}{|c|c|c|c|}
		\hline
		& \bfseries Total & \bfseries Solved & \bfseries Runtime \\
		\hline
		\hline
		2 variables  & $4\,000$ & $4\,000$ & $0.00024 \, s$\\
		\hline
		3 variables  & $4\,560$ & $4\,560$ & $0.00069 \, s$\\
		\hline
		4 variables  & $441$ & $441$ & $0.00084 \, s$\\
		\hline
		5 variables  & $999$ & $999$ & $0.00147 \, s$ \\
		\hline
	\end{tabular}
	\caption{Verification of SiMBA on the dataset of linear MBAs provided by NeuReduce's Github repository~\cite{gh-neureduce}}\label{tab:neureduce}
\end{table}

We see in Table~\ref{tab:neureduce} that the entire dataset could be solved by our algorithm, meaning that each expression was simplified to the corresponding simpler expression also contained in the dataset. Moreover, the results show that the runtime only increases moderately with the number of variables and even MBAs with five variables can be simplified fast.

\subsubsection{Comparison on MBA-Solver's Dataset}

In Table~\ref{tab:1}, we compare SiMBA to MBA-Blast and MBA-Solver on the dataset provided by the latter's Github repository. Here we compare the results of the simplification by either tool with the corresponding simpler expression that is also provided in this dataset. Since all tools simplify the complex MBAs and their corresponding simpler version to exactly the same expressions in all cases, we do not have to verify their equivalence using any SMT solver. To some degree we have to rely on the MBAs' correct generation by the dataset providers because some MBAs are too complex for a verification using any SMT solver without prior simplification. The dataset includes $1\,008$ expressions with two, three or four variables.

\begin{table}[h]
	\ifacm\else
		\centering
	\fi
	\begin{tabular}{|c|c|c|c|}
		\hline
		\multirow{3}{*}{\bfseries Tool} & \multicolumn{3}{|c|}{\bfseries Average runtime} \\
		\cline{2-4}
		& \bfseries 2 variables & \bfseries 3 variables & \bfseries 4 variables \\
		&  ($551$ expr.) &  ($350$ expr.) &  ($107$ expr.) \\
		\hline
		\hline
		MBA-Blast & $0.02501 \, s$ & $0.06726 \, s$ & --- \\
		\hline
		MBA-Solver & $0.00047 \, s$ & $0.00121 \, s$ & $0.14362 \, s$ \\
		\hline
		SiMBA & $0.00024 \, s$ & $0.00116 \, s$ & $0.00257 \, s$ \\
		\hline
	\end{tabular}
	\caption{Comparison of MBA-Blast, MBA-Solver and SiMBA on the dataset of linear MBAs provided by MBA-Solver's Github repository~\cite{gh-mba-solver}}\label{tab:1}
\end{table}

We see that our algorithm runs significantly faster than MBA-Blast and also faster --- especially for four variables --- than MBA-Solver on the average. Unfortunately we cannot run MBA-Blast for four variables since it is lacking any implementation for that case.

\subsubsection{Comparison on Self-Generated MBAs}

Leveraging Corollary~\ref{cor:ext2}, we have generated $1\,000$ different MBAs for each of the following very simple expressions and a varying number of variables:

\begin{alignat*}{2}
	&e_1(x,y) &&= x+y,\\
	&e_2 &&= 49\,374, \\
	&e_3(x) &&=  3\,735\,936\,685\, x + 49\,374,\\
	&e_4(x,y) &&=  3\,735\,936\,685\, (x\uxor y) + 49\,374,\\
	&e_5(x) &&=  3\,735\,936\,685 \cdot \unot x.
\end{alignat*}

We then simplified those expressions with the various tools. Previously we had to modify MBA-Blast as well as MBA-Solver in order to accept our input expressions without impacting their logic or runtime. Specifically we had to adapt the structure of our input expressions in order to use the canonical representation of linear MBAs and we had to adapt the variable names.

Most importantly, we had to make sure that occuring constants are reduced modulo $2^n$, where $n$ is the number of bits. MBA-Blast and MBA-Solver do not care about this number of bits, but this has two serious drawbacks: On the one hand we might not recognize the equality of expressions which are equivalent for a certain number of bits, but not without any modular reduction. On the other hand, the numeric \texttt{NumPy} types internally use \texttt{C} integers and throw an error if an overflow occurs for them.

\begin{table}[h]
	\ifacm\else
		\centering
	\fi
	\begin{tabular}{|c|c|c|}
		\hline
		\multirow{2}{*}{\bfseries Expr.} & \multicolumn{2}{|c|}{\bfseries Average runtime of MBA-Blast} \\
		\cline{2-3}
		 & \bfseries 2 variables & \bfseries 3 variables \\
		\hline
		\hline
		$e_1$ & $0.02695 \, s$ & $0.05366 \, s$  \\
		\hline
		$e_2$ & $0.02291 \, s$ & $0.04833 \, s$ \\
		\hline
		$e_3$ & $0.02851 \, s$ & $0.05753 \, s$ \\
		\hline
		$e_4$ & $0.03602 \, s$ & $0.05699 \, s$ \\
		\hline
		$e_5$ & $0.02782 \, s$ & $0.05459 \, s$  \\
		\hline
	\end{tabular}
	\caption{Runtime of MBA-Blast on 1\,000 MBAs generated for five expressions}\label{tab:mba_blast}
\end{table}

MBA-Blast simplifies, independently of the number of variables, all expressions for $e_1$, $e_2$, $e_3$ and $e_5$ --- partially due to our modifications --- exactly to those simpler expressions while expressions for $e_4(x,y) =  3\,735\,936\,685\, (x\uxor y) + 49\,374$ are simplified to \begin{align*}
3\,735\,936\,685\,x&+3\,735\,936\,685\,y+18\,446\,744\,066\,237\,678\,246\,(x\uand y)\\&+49\,374
\end{align*}
since MBA-Blast only finds simple terms if an expression consists of only one term and this is the linear combination of base expressions. Table~\ref{tab:mba_blast} shows that MBA-Blast has a satisfying and stable runtime for two and three variables while it can, in the available implementation, not be used for a higher number of variables. 

\begin{table}[h]	
	\ifacm\else
		\centering
	\fi
	\begin{tabular}{|c|c|c|c|}
		\hline
		\multirow{2}{*}{\bfseries Expr.} & \multicolumn{3}{|c|}{\bfseries Average runtime of MBA-Solver} \\
		\cline{2-4}
		& \bfseries 2 variables & \bfseries 3 variables & \bfseries 4 variables \\
		\hline
		\hline
		$e_1$ & $0.00074 \, s$ & $0.00091 \, s$ & $0.12761 \, s$  \\
		\hline
		$e_2$ & $0.00052 \, s$ & $0.00092 \, s$ & $0.12352 \, s$ \\
		\hline
		$e_3$ & $0.00026 \, s$ & $0.00150 \, s$ & $0.12802 \, s$ \\
		\hline
		$e_4$ & $0.00041 \, s$ & $0.00160 \, s$ & $0.12683 \, s$ \\
		\hline
		$e_5$ & $0.00072 \, s$ & $0.00134 \, s$ & $0.12871 \, s$  \\
		\hline
	\end{tabular}
	\caption{Runtime of MBA-Solver on 1\,000 MBAs generated for five expressions}\label{tab:mba_solver}
\end{table}

MBA-Solver simplifies, independently of the number of variables, all expressions for $e_2$ and $e_5$ --- again partially due to our modifications --- exactly to those simpler expressions. Unfortunately, due to their basis choice, the simple expression $e(x,y) = x+y$ is missed, and "simplified", depending on the number of variables, to one of the following expressions:
\begin{align*}
	&1\cdot\unot(x\uor\unot y)+1\,(x\uand\unot y)+2\,(x\uand y)\\
	&1\cdot\unot(x\uor (\unot y\uor z))+1\cdot\unot(\unot x\uor (y\uor z))+2\cdot\unot(\unot x\uor (\unot y\uor z))\\ &\hspace{1cm}+1\,(\unot x\uand (y\uand z))+1\,(x\uand (\unot y\uand z))+2\,(x\uand (y\uand z)) \\
	&1\cdot\unot(x\uor (\unot y\uor (z\uor t)))+1\cdot\unot(\unot x\uor (y\uor (z\uor t)))\\
	&\hspace{1cm}+2\cdot\unot(\unot x\uor (\unot y\uor (z\uor t)))+1\cdot\unot(x\uor (\unot y\uor (\unot z\uor t)))\\
	&\hspace{1cm}+1\cdot\unot(\unot x\uor (y\uor (\unot z\uor t)))+2\cdot\unot(\unot x\uor (\unot y\uor (\unot z\uor t)))\\
	&\hspace{1cm}+1\,(\unot x\uand (y\uand (\unot z\uand t)))+1\,(x\uand (\unot y\uand (\unot z\uand t)))\\
	&\hspace{1cm}+2\,(x\uand (y\uand (\unot z\uand t)))+1\,(\unot x\uand (y\uand (z\uand t)))\\
	&\hspace{1cm}+1\,(x\uand (\unot y\uand (z\uand t)))+2\,(x\uand (y\uand (z\uand t)))
\end{align*}

Using Z3, it takes $0.06561$ seconds to verify the equivalence to $e_1(x,y)=x+y$ for the first result, $1.01100$ seconds for the second one and $16.63006$ seconds for the third one.

Even more complex results are obtained for $e_3$ and for $e_4$. It is of course not satisfying to get such complex expressions for very simple ones. However, as already mentioned, if MBA-Solver would try different bases as well --- and especially $(x,y,x\uand y,-1)$ --- it would find a simpler solution. Apart from that, as can be seen in Table~\ref{tab:mba_solver}, MBA-Solver works particularly fast for three variables and much slower for four variables. A possible reason is that the tool has to parse a very large lookup table.\footnote{Further experiments suggest that outsourcing the parsing to a preprocessing step would save about $65$ per cent of the runtime for four variables. Furthermore, omitting the search for a solution using exactly one bitwise expression would make (nearly) the whole lookup table obsolete and save more than $98$ per cent of the runtime in total.} Unfortunately we cannot run the tool for higher numbers of variables since its public implementation is not generic enough and restricted to up to four variables.\footnote{For five variables, we generated a lookup table only containing the $2^5 = 32$ base expressions, disabled the search for a solution using exactly one bitwise expression, which would require the whole lookup table, and performed further adaptions to verify that MBA-Solver has a potential to run fast also for $t=5$.}

\begin{table}[h]
	\ifacm\else
		\centering
	\fi
	\begin{tabular}{|c|c|c|c|}
		\hline
		\multirow{2}{*}{\bfseries Expr.} & \multicolumn{3}{|c|}{\bfseries Average runtime of SiMBA} \\
		\cline{2-4}
		& \bfseries 2 variables & \bfseries 3 variables & \bfseries 4 variables \\
		\hline
		\hline
		$e_1$ & $0.00019 \, s$ & $0.00117 \, s$ & $0.00550 \, s$  \\
		\hline
		$e_2$ & $0.00010 \, s$ & $0.00045 \, s$ & $0.00528 \, s$ \\
		\hline
		$e_3$ & $0.00020 \, s$ & $0.00109 \, s$ & $0.00544 \, s$ \\
		\hline
		$e_4$ & $0.00024 \, s$ & $0.00100 \, s$ & $0.00564 \, s$ \\
		\hline
		$e_5$ & $0.00022 \, s$ & $0.00108 \, s$ & $0.00543 \, s$  \\
		\hline
	\end{tabular}
	\caption{Runtime of SiMBA on 1\,000 MBAs generated for five expressions}\label{tab:our}
\end{table}

For SiMBA, all simplified expressions were exactly equal to the corresponding simple expressions. As Table~\ref{tab:our} suggests, it has a very competitive runtime and outperforms especially MBA-Blast. The algorithms' runtimes for $e_1$ are additionally compared in Figure~\ref{fig:runtimes}; they would be similar for other expressions. Unfortunately we are missing a comparison to MBA-Solver for more than four variables as well as an explanation for MBA-Solver's bad performance for four variables.

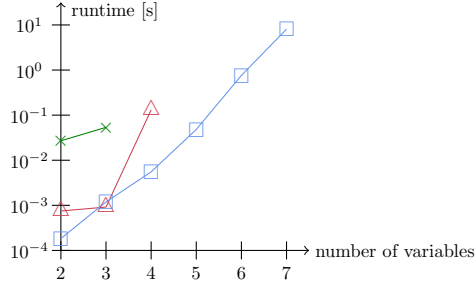
\begin{figure}
	\centering
	\begin{tikzpicture}[scale=0.6]
		\coordinate(blast2) at (2, {log10(0.026950)});
		\coordinate(blast3) at (3, {log10(0.053655)});
		
		\coordinate(solv2) at (2, {log10(0.000743)});
		\coordinate(solv3) at (3, {log10(0.000910)});
		\coordinate(solv4) at (4, {log10(0.127610)});
		
		\coordinate(own2) at (2, {log10(0.000187)});
		\coordinate(own3) at (3, {log10(0.001170)});
		\coordinate(own4) at (4, {log10(0.005502)});
		\coordinate(own5) at (5, {log10(0.048051)});
		\coordinate(own6) at (6, {log10(0.760468)});
		\coordinate(own7) at (7, {log10(8.075448)});
		\coordinate(own8) at (8, {log10(0.005502)});
		\coordinate(own9) at (9, {log10(0.005502)});
		\coordinate(own10) at (10, {log10(0.005502)});
		
		\axeslinlog{2}{7}{-4}{1}{1}{1}
		\node[scale=0.7] at (9.4,-4) {number of variables};
		\node[scale=0.7] at (3.2,1.3) {runtime [s]};
		
		\draw[ao] (blast2) -- (blast3);
		\draw[brickred] (solv2) -- (solv3) -- (solv4);
		
		\draw[cornflowerblue] (own2) -- (own3) -- (own4) -- (own5) -- (own6) -- (own7);
		
		\node[ao, scale=0.8] at (blast2) {$\times$};
		\node[ao, scale=0.8] at (blast3) {$\times$};
		
		\node[brickred, scale=0.8] at (solv2) {$\triangle$};
		\node[brickred, scale=0.8] at (solv3) {$\triangle$};
		\node[brickred, scale=0.8] at (solv4) {$\triangle$};
		
		\node[cornflowerblue, scale=0.8] at (own2) {$\square$};
		\node[cornflowerblue, scale=0.8] at (own3) {$\square$};
		\node[cornflowerblue, scale=0.8] at (own4) {$\square$};
		\node[cornflowerblue, scale=0.8] at (own5) {$\square$};
		\node[cornflowerblue, scale=0.8] at (own6) {$\square$};
		\node[cornflowerblue, scale=0.8] at (own7) {$\square$};
	\end{tikzpicture}
	\caption{Comparison of the runtimes of MBA-Blast (\textcolor{ao}{$\times$}), MBA-Solver (\textcolor{brickred}{$\triangle$}) and SiMBA (\textcolor{cornflowerblue}{$\square$}) on a logarithmic scale for MBAs for $e_1$}\label{fig:runtimes}
\end{figure}

\subsubsection{Encoded MBAs}

In order to make MBAs in general harder to solve, its inputs and/or its output can be encoded using specific functions. We provide some experiments on MBAs whose outputs are encoded using affine functions $f: B^n \to B^n$, $f(x) = ax+b$ where $a,b \in B^n$ are randomly determined. Note that MBA-Blast and MBA-Solver cannot deal with this encoding since everything would have to be multiplied out in order to have the canonical representation as stated in Definition~\ref{def:linear}, and that input encoding would cause $e_4$ and $e_5$ to be nonpolynomial.

In Table~\ref{tab:encoded}, we consider runtimes of the simplification of encoded versions of the expressions used in the previous section, i.e., $$e_1^\prime(x,y) = a(x+y) + b\ \text{for random } a,b \in B^{64}$$ and equivalently for $e_2^\prime,\ldots,e_5^\prime.$

\begin{table}[h]
	\ifacm\else
		\centering
	\fi
	\begin{tabular}{|c|c|c|c|}
		\hline
		\multirow{2}{*}{\bfseries Expr.} & \multicolumn{3}{|c|}{\bfseries Average runtime of SiMBA} \\
		\cline{2-4}
		& \bfseries 2 variables & \bfseries 3 variables & \bfseries 4 variables \\
		\hline
		\hline
		$e_1^\prime$ & $0.00020 \, s$ & $0.00118 \, s$ & $0.00484 \, s$  \\
		\hline
		$e_2^\prime$ & $0.00010 \, s$ & $0.00029 \, s$ & $0.00481 \, s$ \\
		\hline
		$e_3^\prime$ & $0.00021 \, s$ & $0.00125 \, s$ & $0.00460 \, s$ \\
		\hline
		$e_4^\prime$ & $0.00020 \, s$ & $0.00098 \, s$ & $0.00516 \, s$ \\
		\hline
		$e_5^\prime$ & $0.00015 \, s$ & $0.00110 \, s$ & $0.00547 \, s$  \\
		\hline
	\end{tabular}
	\caption{Runtime of SiMBA on 1\,000 MBAs with output encoding generated for five expressions}\label{tab:encoded}
\end{table}

As expected, the runtimes do not significantly deviate from that for the non-encoded MBAs. Each encoded MBA has been simplified to the same expression as its simpler equivalent, which is of course encoded with the very same function.

\section{Conclusion}

The algorithm described in this paper is designed to simplify linear MBAs, which clearly is a restriction. Regarding the success at simplification of linear MBAs and the runtime taken, it is more than competitive with comparable tools. 

MBA-Blast and MBA-Solver are efficient tools which are based on a groundbreaking transformation between $B=\{0,1\}$ and $B^n$ for any $n\in\NN$ that dates back to 2007~\cite{zhou}. In the publicly available implementation, MBA-Solver uses a lookup table in all cases and hence cannot be generalized to arbitrary variable counts. 

We see that MBA-Solver's runtime increases drastically with four variables, but we cannot observe more of a trend since we cannot run it for higher numbers of variables. It uses an unhandy basis for bitwise expressions, and although that can, according to the paper~\cite{mba-solver}, be changed, solving the equation system would cause much more effort then: While MBA-Blast uses \texttt{NumPy}'s \texttt{linalg.solve} function for solving it, MBA-Solver leverages the fact that all basis elements have exactly one $1$ in their truth vectors and hence the equation system is trivially solved.

We summarize the main characteristics of SiMBA as follows:
\begin{enumerate}
	\item It does not require any specific structure of input expressions as long as they can be rewritten as linear MBAs.
	
	\item It checks whether input expressions are indeed linear.
	
	\item It does not require any specific notation of the variables nor any declaration of them, but parses them automatically.
	
	\item It aims for nicest solutions in all relevant cases for expressions with at most three variables. In more detail, it does so for expressions whose output depends on not more than three variables, i.e., after dropping unnecessary variables.
	
	\item It is implemented to work for an arbitrary number of variables.
	
	\item It can deal with arbitrarily high constants which are constantly reduced modulo $2^n$ for the word length $n\in\NN$.
	
	\item It does not use any \texttt{Python} package for simplification such as \texttt{SymPy} and no package such as \texttt{NumPy} for the solution of the linear equation system.
	
	\item As our experiments suggest, it has a very good performance as compared to other existing tools.
	
	\item It can handle linear MBAs whose outputs are encoded using affine functions.
	
	\item It can be extended in a straightforward way to the simplification of polynomial and nonpolynomial MBAs.
\end{enumerate}

We particularly want to point out that one main application of MBAs is given by \textit{opaque predicates}, i.e., predicates which constantly evaluate to either $0$ or $1$, and SiMBA can easily simplify all linear MBAs used in opaque predicates to these constants.

As a major contribution, we proved that a linear MBA can be written as a linear combination of base bitwise expressions via a direct evaluation for tuples of zeros and ones and hence a transformation from $B^n$, for $n\in\NN$, to $B=\{0,1\}$ is not necessary --- see Theorem~\ref{thm:main}.

\section{Outlook to Nonlinear Mixed Boolean-Arithmetic Expressions}

As mentioned previously, MBA-Blast and MBA-Solver are also applicable to polynomial MBAs. Just like those, SiMBA is extendable to this class of MBAs and to nonpolynomial MBAs as well after identifying all linear subexpressions of a nonlinear MBA. Unfortunately, it seems that a global straightforward approach via evaluation of nonlinear MBAs on a small set of input values is not possible.

Table~\ref{tab:teaser} shows a comparison of an adaption of SiMBA with MBA-Solver for the simplification of linear, polynomial and nonpolynomial MBAs provided by MBA-Solver's Github repository~\cite{gh-mba-solver}. While both tools can, in combination with Z3, verify all MBAs' equivalence to their corresponding simpler expressions, SiMBA simplifies all MBAs exactly equivalently as their simpler expressions. Additionally it simplifies $44$ of $104$ additional nonlinear MBAs which are marked unsolvable by MBA-Solver exactly as their simpler expressions, and it outperforms MBA-Solver regarding runtime in all categories. In contrast to MBA-Solver, SiMBA does not require any information whether an MBA is polynomial or even linear, and it does not use the additional hints about linear subexpressions provided for nonpolynomial MBAs.

\begin{table}[h]
	\ifacm\else
		\centering
	\fi
	\begin{tabular}{|c|c|c|c|c|c|}
		\hline
		\multirow{2}{*}{\bfseries Category} & \multirow{2}{*}{\bfseries Total} & \multicolumn{2}{|c|}{\bfseries MBA-Solver} & \multicolumn{2}{|c|}{\bfseries SiMBA} \\
		\cline{3-6}
		 & & \bfseries Solved & \bfseries Runtime  & \bfseries Solved & \bfseries Runtime \\
		\hline
		\hline
		Linear  & $1\,008$ & $1\,008$ & $0.01546 \, s$ & $1\,008$ & $0.00250 \, s$\\
		\hline
		Polynom.  & $1\,008$ & $1\,008$ & $0.02271 \, s$ & $1\,008$ & $0.00326 \, s$\\
		\hline
		Nonpolyn.  & $899$ & $109$ & $0.07965 \, s$  & $899$ & $0.00957 \, s$\\
		\hline
	\end{tabular}
	\caption{Comparison of a simple adaption of SiMBA with MBA-Solver on the datasets of linear, polynomial and nonpolynomial MBAs provided by MBA-Solver's Github repository~\cite{gh-mba-solver}}\label{tab:teaser}
\end{table}

In contrast to the algorithm described in Section~\ref{subsec:our}, SiMBA's adaption parses an input expression into an acyclic syntax tree (AST), identifies and simplifies linear subexpressions and  applies simple additional tricks.

%

\ifacm
\else
	\bibliographystyle{plain}
\fi

\bibliography{refs}



\end{document}